\documentclass[12pt]{article}
\pdfoutput=1

\usepackage{geometry}
\usepackage{amsfonts,amsmath,amssymb,amsthm}
\usepackage[utf8]{inputenc}
\usepackage[svgnames]{xcolor}
\usepackage[colorlinks=true,linkcolor=DarkSlateBlue, citecolor=DarkGreen,urlcolor=DarkRed]{hyperref}
\usepackage{graphicx}
\usepackage[export]{adjustbox}
\usepackage[superscript]{cite}
\usepackage{bm}

 \usepackage{titlesec}
 \titleformat*{\subsection}{\bfseries\boldmath}
 \titleformat*{\section}{\large\bfseries\boldmath}

\newcommand\doi[2]        {\href{https://dx.doi.org/#1}{#2}}

\newtheorem{proposition}{Proposition}[section]
\theoremstyle{definition}
\newtheorem{definition}{Definition}[section]

\begin{document}

\thispagestyle{empty}
\phantom{.}
\vspace{.7cm}
\begin{center}{\Large \textbf{
Translation invariant defects
\\[.5em]
as an extension of topological symmetries}}
\end{center}

\vspace{.5em}

\begin{center}
\large
Federico Ambrosino$\,{}^{(1,2)}$,
Ingo Runkel$\,{}^{(1,3)}$ and  G\'erard M.\ T.\ Watts$\,{}^{(1,5)}$
\end{center}

\vspace{.5em}

\begin{center}{\small
\begin{tabular}{ll}
$^{(1,2)}$  & Deutsches Elektronen-Synchrotron DESY,
Notkestr.\ 85, 22607 Hamburg, Germany,\\
& and Perimeter Inst.\ for Theoretical Physics, Waterloo, Ontario N2L 2Y5, Canada
\\[.2em]
& {\sf federicoambrosino25@gmail.com}
\\[.5em]
$^{(1,3)}$ & Fachb.\ Mathematik, Universit\"at Hamburg,
Bundesstr.\ 55, 20146 Hamburg, Germany
\\[.2em]
& {\sf ingo.runkel@uni-hamburg.de}
\\[.5em]
$^{(1,5)}$ & Dept.\ of Mathematics,
	King’s College London, Strand, London WC2R 2LS, UK
\\[.2em]
& {\sf gerard.watts@kcl.ac.uk}
\end{tabular}
}
\end{center}

\vspace{1em}

\noindent
\textbf{Abstract:}
The modern way to understand symmetries of a quantum field theory is via its topological defects in various dimensions. In this contribution to the proceedings we focus on line defects in 2d\,QFT and we point out that topological defects naturally embed into a larger class, namely translation invariant defects. The latter still allow for non-singular fusion and one obtains a monoidal category of translation invariant defects which contains that of topological defects as a full subcategory. We give a simple perturbative description of translation invariant defects in a perturbed conformal field theory via chiral three-dimensional topological field theory. We show in the example of the Ising CFT and the Lee-Yang CFT that even if no topological defects survive the deformation, some translation invariant defects still do.

\vspace{1em}

\tableofcontents

\newpage

\noindent
A $d$-dimensional quantum field theory (QFT) has observables localised on submanifolds of all dimensions $0,1,\dots,d$. Of these, zero-dimensional observables (fields) and one-dimensional observables (e.g.\ Wilson lines in gauge theory) are perhaps the most familiar. In this contribution to the proceedings, we will focus on two-dimensional QFT and a special kind of one-dimensional observables, so-called translation invariant defects, as well as their properties and applications. These proceedings are mainly based on Ref.\ \citen{Ambrosino:2025myh}.

\medskip

Before discussing translation invariant defects in more detail, we will have a closer look at the special case of topological defects.

\section{Topological line defects in two-dimensional field theory}\label{sec:topdef}

A line defect in a two-dimensional quantum field theory is \emph{topological} if it can be deformed without affecting the value of correlators, provided it is not deformed across field insertions or other defect lines.

Topological defects can be composed by fusion of parallel line defects, and they form a tensor category whose objects are the different topological defect conditions, the tensor product is the fusion operation, and the morphisms are topological junction fields\cite{Fuchs:2002cm,Davydov:2011kb,Chang:2018iay}. This tensor category describes the \textit{topological symmetry} of the 2d\,QFT and generalizes the usual case of a symmetry group in two ways: firstly, topological defects need not be invertible, and secondly, the topological junctions are additional structure even in the group case (where they define a 3-cocycle in group cohomology that is an obstruction to gauging the symmetry\cite{Frohlich:2009gb,Bhardwaj:2017xup}).

\medskip

Non-conformal quantum field theories generically are part of a 1-parameter family under scale transformation (the renormalisation group (RG) flow), with limit points being conformal field theories at small scales (the ultraviolet (UV) CFT) and at large scales (the infrared (IR) CFT).
The IR CFT may be trivial and just consist of one or more vacuum states of a massive QFT.

\begin{figure}

\centering{\includegraphics[scale = 1]{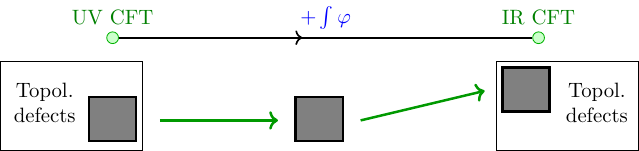}}

\caption{Topological defects along a renormalisation group flow. Only some of the topological defects (gray box) of the UV CFT are preserved under the perturbation by the bulk field $\varphi$. These flow to some of the topological defects of the IR CFT.}
\label{fig:top-def-RG}
\end{figure}

The UV and IR CFT will typically have a much larger topological symmetry than the 1-parameter family of QFTs connecting them. Intuitively, the topological symmetry of a CFT has to be large enough to accommodate that of all RG flows that can start or end at it. Thus we expect a picture like in Figure~\ref{fig:top-def-RG}.
In particular, defects in the UV CFT which remain topological along the RG flow are part of the topological symmetry of the IR CFT. Since in general it is a difficult problem to determine the RG endpoint starting from a UV CFT and a relevant perturbing field, any accessible information about the IR CFT is helpful.

\medskip

A very useful construction is to define a three-dimensional topological field theory in terms of the tensor category of topological defects via a state-sum construction (technically this requires one to work with a spherical fusion category of topological defects), and to couple this 3d\,TFT to the 2d\,QFT so that the QFT becomes a non-topological boundary condition of the 3d\,TFT, see Ref.\ \citen{Gaiotto:2020iye} and e.g.\ the review in Ref.\ \citen{saclay-proc}.
The original 2d\,QFT can be recovered by considering a surface times an interval with the non-topological boundary condition on one side, and a specific topological boundary condition on the other, see Figure~\ref{fig:three-bulk-fields}b for an illustration.
The same idea has proven very powerful also in higher-dimensional QFTs and in condensed-matter contexts, and the TFT is called \emph{symmetry topological field theory} (SymTFT) \cite{Kong:2020cie,Apruzzi:2021nmk,Freed:2022qnc,Kaidi:2022cpf,Bhardwaj:2023ayw}.

\begin{figure}

\centering{\includegraphics[scale = 0.9]{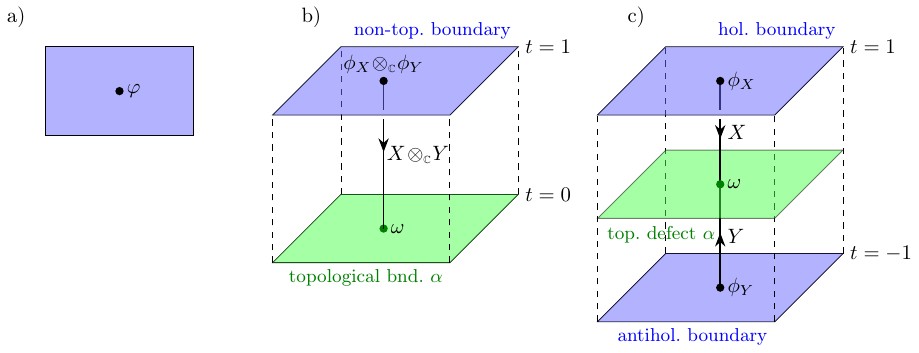}}

\caption{a) A patch of the CFT world sheet $\Sigma$ with an insertion of a bulk field $\varphi$. b) The same patch in the SymTFT representation on the three-manifold $\Sigma \times [0,1]$. c) Chiral TFT representation on the three-manifold $\Sigma \times [-1,1]$.}
\label{fig:three-bulk-fields}
\end{figure}

Instead of SymTFT, below we will use the so-called chiral TFT which is well-suited to 2d\,CFTs \cite{Fuchs:2002cm}. In Figure~\ref{fig:three-bulk-fields} we compare how to write a bulk field $\varphi$ of the 2d\,CFT in SymTFT and chiral TFT language.

In the SymTFT representation (Figure~\ref{fig:three-bulk-fields}b), the upper boundary is the non-topological boundary and the bottom boundary is topological. Suppose the bulk field transforms in the representation $X \otimes_\mathbb{C} Y$, where $X$ is a representation of the holomorphic chiral algebra and $Y$ of the anti-holomorphic one. Then the bulk field is described by a topological line defect labelled $X \otimes_\mathbb{C} Y$, an element $ \phi_X \otimes_\mathbb{C} \phi_Y \in X \otimes_\mathbb{C} Y$, and a topological junction $\omega$ of the line defect with the topological boundary.

In the chiral TFT description in Figure~\ref{fig:three-bulk-fields}c, both boundaries of the three-manifold are non-topological, with $\Sigma \times \{1\}$ carrying the holomorphic degrees of freedom, and $\Sigma \times \{-1\}$ the anti-holomorphic ones. At $\Sigma \times \{0\}$ a topological surface defect is placed which controls how holomorphic and anti-holomorphic fields are combined. The bulk field $\varphi$ is represented via elements $\phi_X \in X$, $\phi_Y \in Y$, a line defect labelled $X$ from the holomorphic boundary to the surface defect, a line defect $Y$ starting at the anti-holomorphic boundary, and a junction $\omega$ where $X$ and $Y$ meet on the surface defect $\alpha$.

One passes from the chiral TFT in Figure~\ref{fig:three-bulk-fields}c to the SymTFT in Figure~\ref{fig:three-bulk-fields}b by folding \cite{Carqueville:2023jhb}. Below we will exclusively work in the chiral TFT picture as it makes the separation into holomorphic and anti-holomorphic coordinate dependence manifest.

\medskip

Consider a 2d\,CFT and a topological line defect $D$ of that CFT. If we perturb the CFT by a bulk field $\varphi$, we can ask if $D$ remains topological also in the perturbed theory. The condition for this is simply that $D$ should commute with the perturbing field\cite{Chang:2018iay,Fredenhagen:2009tn}:
\begin{equation}\label{eq:topdef-preserved-ws}
\includegraphics[valign = c,scale=0.75]{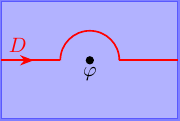} \quad =\quad  \includegraphics[valign = c, scale=0.75]{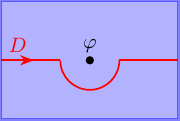}
\end{equation}
The figure shows a patch of the CFT world sheet $\Sigma$, and the identity is understood to hold inside correlators, provided the field and defect configuration is the same outside the patch shown.
In the chiral TFT presentation, this condition looks as follows:
\begin{equation}\label{eq:topdef-preserved-chiralTFT}
\includegraphics[valign = c,scale=0.75] {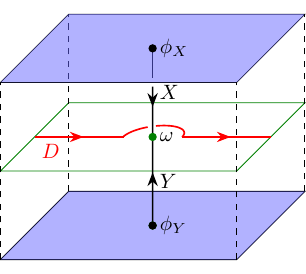} \quad = \quad \includegraphics[valign = c,scale=0.75] {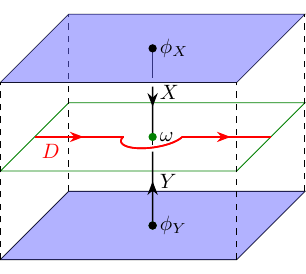}
\end{equation}
The description of the bulk field is as in Figure~\ref{fig:three-bulk-fields}c. The line defect $D$ is represented as a line defect on the topological surface defect $\alpha$.

We stress that the question whether the bulk field $\varphi$ commutes with the topological defect $D$ is now entirely phrased within the chiral TFT in terms of properties of its various topological defects.

\subsubsection*{Example: Ising model}

The 2d Ising CFT has central charge $c=\frac12$. In the chiral TFT picture, there are three elementary line defects labelled $\mathbf{1}$, $\sigma$, $\epsilon$, corresponding to the three irreducible representations $R_h$ of the irreducible Virasoro vertex operator algebra (VOA) at $c=\frac12$ with lowest $L_0$-weights $h_\mathbf{1}=0$, $h_\sigma=\frac1{16}$, and $h_\epsilon = \frac12$.

The only available surface defect turns out to be the trivial one, and so line defects on the surface defect are just the same as line defects of the chiral TFT. Accordingly, the elementary topological defects of the Ising CFT are labelled by $\mathbf{1}$, $\sigma$ and $\epsilon$, with fusion rule $\sigma \otimes \sigma = \mathbf{1} \oplus \epsilon$.

Two elementary line defects $X,Y$ have a non-trivial topological junction with the trival surface defect iff $X=Y$. Hence we find three primary bulk fields which we denote by bold symbols,
\begin{align}
\bm{1} &= |0\rangle \otimes_\mathbb{C} |0\rangle \in R_0 \otimes_\mathbb{C} R_0
\nonumber\\
\bm{\sigma} &= |\tfrac1{16}\rangle \otimes_\mathbb{C} |\tfrac1{16}\rangle \in R_{\frac1{16}} \otimes_\mathbb{C} R_{\frac1{16}}
\\
\bm{\epsilon} &= |\tfrac12\rangle \otimes_\mathbb{C} |\tfrac12\rangle \in R_{\frac12} \otimes_\mathbb{C} R_{\frac12} \ .
\nonumber
\end{align}
Here $|h\rangle \in R_h$ denotes the primary state (i.e.\ the of lowest $L_0$-weight).
The following table lists which elementary topological line defects commute with each of the bulk fields:

\begin{center}
\begin{tabular}{r|c|c}
perturbing bulk field & $\bm{\epsilon}$ & $\bm{\sigma}$
\\
\hline
conserved topological defects & $\mathbf{1}$, $\mathbf{\epsilon}$
& $\mathbf{1}$
\end{tabular}
\end{center}

Thus, if we perturb the Ising CFT by the bulk field $\bm{\epsilon}$ (the temperature perturbation), the identity defect and the $\epsilon$-defect remain topological, giving a $\mathbb{Z}_2$-symmetry which is conserved along the RG flow. This corresponds to the spin-flip symmetry of the Ising lattice model in zero magnetic field. The $\sigma$-defect is not conserved but instead changes the sign of the perturbation when one commutes it past the perturbing field $\epsilon$, implementing high-low temperature duality \cite{Frohlich:2004ef}.

For the $\bm{\sigma}$-perturbation (magnetic field perturbation), only the identity defect is preserved.
The $\epsilon$ defect is not preserved but changes the sign of the $\bm{\sigma}$-perturbation.

\subsubsection*{Example: Lee-Yang model}

The Lee-Yang model is the non-unitary minimal model CFT with central charge $c=-22/5$. It is built from two Virasoro representations $R_h$ with $h=0$ and $h=-1/5$. We write $\tau = R_{-1/5}$ and $\bm{\tau} = |{-}\frac1{5}\rangle \otimes_\mathbb{C} |{-}\frac1{5}\rangle$. Again there is no non-trivial surface defect and the line defects are labelled $\mathbf{1}$ and $\tau$ with fusion rule $\tau \otimes \tau = \mathbf{1} \oplus \tau$. The $\tau$ defect does not commute with the bulk field $\bm{\tau}$ and hence only the identity defect is preserved if we perturb the Lee-Yang CFT by the bulk field $\bm{\tau}$.

\medskip

In the next section we will see that if we pass to the larger class of translation invariant defects, of which topological ones are a special case, we do find defects compatible with the $\bm{\sigma}$-perturbation of the Ising CFT and with the $\bm{\tau}$-perturbation of the Lee-Yang CFT.

\section{Translation invariant line defects}

A line defect $D$ is called \textit{translation invariant} if its defect operator $\widehat D$ commutes with the Hamiltonian. In more detail, consider a 2d\,QFT on a cylinder of circumference $L$ and place the line defect $D$ along the periodic direction. This produces the defect operator $\widehat D : \mathcal{H} \to \mathcal{H}$ on the state space $\mathcal{H}$ of the QFT on a circle of circumference $L$. Denoting by $H(L)$ the Hamiltonian acting along the cylinder, the condition is
\begin{equation}
	\big[H(L),\widehat{D}\big]=0 \ .
\end{equation}
Such defects have been investigated for example in Refs.\,\citen{Bazhanov:1994ft,Konik:1997gx,Bachas:2004sy,Runkel:2007wd}.

In general, the process of fusing two line defects is singular \cite{Bachas:2007td,Bachas:2013ora}.
For translation invariant defects, however, this fusion is non-singular as the overall operator given by two parallel line defects on a cylinder with distance $r$ is independent of $r$.
Accordingly, translation invariant defects will again form a tensor category \cite{Ambrosino:2025myh}, which we denote by $\mathcal{T}$. The objects of $\mathcal{T}$ are translation invariant defects, and the morphisms are topological point junctions (in the sense that they can be translated freely along the defect, as well as with the defect).
The tensor category $\mathcal{T}_\text{top}$ of topological defects is a full subcategory
\begin{equation}\label{eq:Ttop-in-T}
		\mathcal{T}_\text{top} \subset \mathcal{T} \ .
\end{equation}
Indeed, morphisms between two topological defects were also defined to be topological point junctions.

The main point of these proceedings and of Ref.\ \citen{Ambrosino:2025myh} is to propose that it is useful to \textit{think of $\mathcal{T}$ as an extension of the topological symmetry $\mathcal{T}_\text{top}$} which shares some of its good properties, like the existence of a well-defined fusion operation, and which can provide useful information about a QFT, for example on its renormalisation group behaviour.

\medskip

If the QFT in question is a CFT, it makes sense to talk about conformal defects. A defect which is both conformal and translation invariant is actually topological \cite{Ambrosino:2025myh}. This is an easy consequence of the commutators with the Virasoro modes (after mapping the cylinder to the  plane):
\begin{align}
	\text{$D$ is conformal}
	~&:~ [L_n - \overline L_{-n},\widehat{D}] = 0
	\text{ for all } n \in \mathbb{Z} \ ,
\nonumber	\\
	\text{$D$ is translation invariant}
	~&:~ [L_0 + \overline L_{0},\widehat{D}] = 0 \ .
\end{align}
The two conditions together imply that $L_0$ and $\overline L_0$ commute separately with $\widehat{D}$, and combining this with the first condition one quickly checks that all $L_n$ and $\overline L_n$ individually commute with $\widehat{D}$.

\begin{figure}

\centering{\includegraphics[scale = 1]{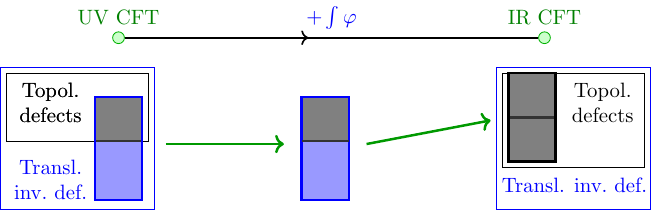}}

\caption{Translation invariant defects along a renormalisation group flow. The translation invariant defects compatible with the bulk perturbation flow to topological defects of the IR CFT.}
\label{fig:transinv-def-RG}
\end{figure}

An important consequence is the following simple observation: \textit{Under the renormalisation group, a translation invariant defect flows to a topological defect}.
This applies to both, defect flows when the bulk CFT remains unperturbed, and combined bulk and defect flows.
In particular, if we include translation invariant defects, we can update Figure~\ref{fig:top-def-RG} to Figure~\ref{fig:transinv-def-RG}.
By including translation invariant defects, one may thus hope to see more of the topological symmetry of the IR CFT than if one would restricting oneself to topological defects.

\medskip

There is a simple sufficient condition for a perturbed defect to be translation invariant in the perturbed (or unperturbed) CFT to all orders in the coupling constant. This is the topic of the next section.

\section{The commutation condition}\label{sec:commcond}

Let the 2d\,CFT be described by a surface defect $\alpha$ of the chiral TFT, and let $\varphi = \phi_X \otimes_\mathbb{C} \phi_Y$ be a bulk field for the line defects $X,Y$ and junction $\omega$ as in Figure~\ref{fig:three-bulk-fields}c. The case $\varphi=0$ is allowed and corresponds to $\omega$ being zero.

The Hamiltonian of the perturbed CFT on a cylinder of circumference $L$ is
\begin{equation}\label{eq:H0+Hpert}
	H(\mu) = H_0 + H_\mathrm{pert}(\mu) \ ,
\end{equation}
where
\begin{equation}\label{eq:Hpert}
	H_0 = \frac{2\pi}{L} \Big( L_0 + \overline L_0 - \frac{c}{12} \Big)
	~,\quad
	H_\mathrm{pert}(\mu) = 2i \mu \int_0^L \hspace{-.4em} \varphi(s) \, ds  \ .
\end{equation}
The factor of $2i$ in $H_\mathrm{pert}(\mu)$ is a convention and results in fewer factors of $i$ later on. Having both $\mu$ and $\omega$ is redundant: to obtain $\mu \varphi$ one can replace $\omega$ by $\mu\omega$. We will still keep $\mu$ as it helps to separate the different orders in the perturbative expansion.

We will use a very specific form of perturbation for the topological line defects. Namely let $D$ label a topological defect (which can be elementary or a sum of elementary defects). On $D$ we consider a holomorphic defect field $\psi$ and an antiholomorphic defect field $\bar\psi$ which are represented in chiral TFT as:
\begin{equation}
\includegraphics[valign = c, scale=0.75]{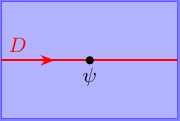} ~=~ \includegraphics[valign = c, scale=0.75]{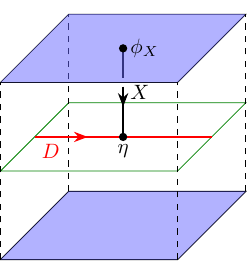}
\qquad \qquad
\includegraphics[valign = c, scale=0.75]{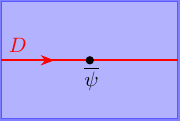} ~=~ \includegraphics[valign = c, scale=0.75]{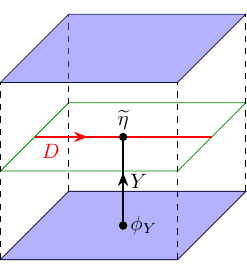}
\end{equation}
Let us look at $\psi$ in more detail, the description of $\bar\psi$ is analogous. Firstly, $D$ is a line defect on the surface defect $\alpha$ of the chiral TFT. Next, $X$ is the same line defect as used to define the perturbing bulk field $\varphi$ in Figure~\ref{fig:three-bulk-fields}c, and $\phi_X \in X$ is the same vector as used in the definition of $\varphi$. Finally, $\eta$ is a topological point junction on the surface defect $\alpha$, joining $X$ to $D$. In terms of the 2d\,CFT, $\psi$ is a holomorphic defect field on the topological defect $D$ of $(L_0,\bar L_0)$-weights $(h_{\phi_X},0)$. Analogously $\bar\psi$ is an antiholomorphic defect field with weights $(0,h_{\phi_Y})$.
It is allowed for $\psi$ or $\bar\psi$ (or both) to be zero, which is implemented by $\eta$ or $\tilde\eta$ being zero.

We say that the fields $\psi$, $\bar\psi$ satisfy the \textit{commutation condition for $\varphi$} if the following identity holds in the chiral TFT description \cite{Runkel:2010ym,Buecher:2012ma,Ambrosino:2025myh}:
\begin{equation}\label{eq:commcond}
\includegraphics[scale=0.75,valign = c]{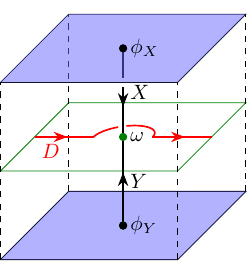} ~-~ \includegraphics[scale=0.75,valign = c]{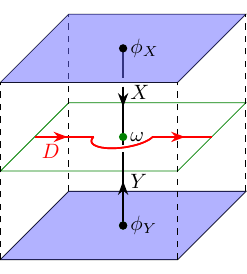} ~~=~~ \includegraphics[scale=0.75,valign = c]{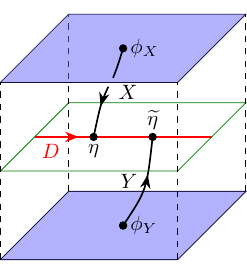} ~-~ \includegraphics[scale=0.75,valign = c]{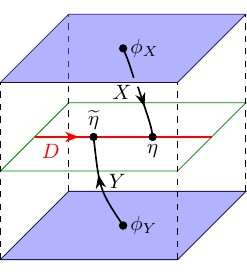}
\end{equation}

To explain the relevance to perturbed defects, let us write
\begin{equation}
	D(\lambda \psi + \tilde\lambda \bar\psi)
\end{equation}
for the topological defect $D$ perturbed by $\lambda \psi + \tilde\lambda \bar\psi$. This amounts to inserting the exponentiated integral $\exp\int_0^L \big(\lambda \psi(s) + \tilde\lambda \bar\psi(s)\big) ds$ on the defect $D$. It is redundant to have both $\lambda,\tilde\lambda$ and $\eta,\tilde\eta$ as we could instead have used $\lambda\eta$ and $\tilde\lambda \tilde\eta$. As for $\mu$ in the bulk perturbation \eqref{eq:Hpert}, we keep the parameters $\lambda,\tilde\lambda$ to help organise the perturbative expansion.

We assume that each individual multiple integral in the expansion of the exponential does not require regularisation, i.e.\ that the singularities occurring in the iterated operator product expansions of $\psi$'s and $\bar\psi$'s are mild enough so that one can integrate over them. Let $h_\text{min}$ be the minimal conformal weight of a holomorphic field generated by the OPEs of $\psi$ with itself. For a unitary CFT this would be the identity field  ($h_\text{min}=0$), but for non-unitary theories one can have $h_\text{min}<0$. Then a necessary condition for regularity of the integrals is
\begin{equation}\label{eq:regularity}
	2h_{\phi_X}-h_\text{min}<1 \ ,
\end{equation}
and analogously for $\bar\psi$.
For a more detailed discussion of singularities of iterated OPEs see e.g.\ Ref.\ \citen{Klassen:1990dx}.

If the commutation condition \eqref{eq:commcond} and the regularity condition \eqref{eq:regularity}
are satisfied, one can show that $D(\lambda \psi + \tilde\lambda \bar\psi)$ is a translation invariant defect in the 2d\,CFT perturbed by $2i\mu\varphi$ to all orders in $\lambda$, $\tilde\lambda$ for the choice of bulk coupling constant given by $\mu = \lambda \tilde\lambda$. In terms of the defect operator $\widehat D(\lambda \psi + \tilde\lambda \bar\psi)$ of the perturbed defect, this means that\cite{Runkel:2010ym,Buecher:2012ma}
\begin{equation}\label{eq:def-comm-ham}
		\big[ \,H(\mu) \,,\, \widehat D(\lambda \psi + \tilde\lambda \bar\psi) \, \big] \,=\, 0
		\quad \text{for}~~\mu=\lambda\tilde\lambda
\end{equation}
to all orders in $\lambda$, $\tilde\lambda$. The condition $\mu=\lambda\tilde\lambda$ is absent if any of $\omega$, $\eta$, $\tilde\eta$ are zero.

\medskip

Let us look at two special cases of the commutation condition \eqref{eq:commcond}.

\medskip

\noindent
1) The first case is that $\eta=0=\tilde\eta$. Then $\psi=0=\bar\psi$ and the topological defect is actually unperturbed. The right hand side of the commutation condition is  zero, so that the condition now is that the topological defect $D$ commutes with the perturbing field $\varphi$ as in \eqref{eq:topdef-preserved-chiralTFT}. Hence this special case covers topological defects $D$ that remain topological in the theory perturbed by $\varphi$.

\medskip

\noindent
2) The second case is $\omega=0$ (or equivalently $\varphi=0$), and the bulk CFT is unperturbed. In this case, the left hand side of the commutation condition is zero, and it now requires the two defect fields $\psi$ and $\bar\psi$ to commute with each other. This is automatically the case if $\tilde\eta=0$, or, equivalently, $\bar\psi=0$. The defect $D$ is then perturbed only by the holomorphic defect field $\psi$ and indeed it is easy to see that such a perturbed defect is always translation invariant in the unperturbed bulk CFT \cite{Bazhanov:1994ft,Konik:1997gx,Bachas:2004sy,Runkel:2007wd}. The same holds if $\eta=0$ (i.e.\ $\psi=0$) so that the defect is only perturbed by $\bar\psi$. The mixed situation with both $\psi$ and $\bar\psi$ nonzero occurs for example when one fuses a defect perturbed by only $\psi$ with one perturbed by only $\bar\psi$ as described in the next section.

Recall that a translation invariant defect flows to a topological defect. In the special case $\omega=0$ this means that if we start from a topological defect $D$ in our 2d\,CFT, perturbing $D$ by a commuting pair $\psi$, $\bar\psi$ generates a flow to another topological defect of the same CFT. For a general perturbation this is typically not the case \cite{Kormos:2009sk,Popov:2025cha}.

\medskip

In the general case with $\omega \neq 0$, the combined bulk and defect flow is as in Figure~\ref{fig:transinv-def-RG}.

\subsubsection*{Example: Ising 2d\,CFT}

We consider the perturbation of the Ising CFT by the bulk field $\bm{\epsilon}$ (temperature perturbation) and $\bm{\sigma}$ (magnetic field perturbation) in turn. In both cases we will see that there are more translation invariant defects compatible with the perturbation than there are topological ones.

\medskip

\noindent
\textit{$\bm{\epsilon}$-perturbation:} In the notation in Figure~\ref{fig:three-bulk-fields}c we have $X=\epsilon=Y$, $\phi_X = |\frac12 \rangle = \phi_Y$. The space of topological junctions $\omega$ is one-dimensional, and we fix a non-zero $\omega$; this determines the normalisation of the bulk field $\bm{\epsilon}$.

We saw in Section~\ref{sec:topdef} that the $\epsilon$-defect commutes with $\bm{\epsilon}$, that is, it satisfies the commutation condition with $\eta=0=\tilde\eta$. The topological $\sigma$-defect does not commute with $\bm{\epsilon}$, but there is a solution to the commutation condition with non-zero $\eta$, $\tilde\eta$. However, in this case the regularity condition \eqref{eq:regularity} is not satisfied as we have $h=\frac12$ and $h_\text{min}=0$. The integrals entering the perturbed defect $D_\sigma(\lambda \psi + \tilde\lambda\bar\psi)$ need to be regularised and one needs to check separately that there is regularisation which does not break translation invariance. We return to this point briefly in the next section.

\medskip

\noindent
\textit{$\bm{\sigma}$-perturbation:} In Figure~\ref{fig:three-bulk-fields}c we set $X=\sigma=Y$, $\phi_X = |\frac1{16} \rangle = \phi_Y$, and $\omega \neq 0$. The space of topological junctions is again one-dimensional and we fix a choice of $\omega$. The commutation condition has no solution with $\eta=0=\tilde\eta$ other than the $\mathbf{1}$-defect, that is, no non-trivial topological symmetry is preserved along the flow. However, for the non-elementary defect $D = \mathbf{1} \oplus \sigma$, the commutation condition does have a solution for some non-zero $\eta$, $\tilde\eta$. The regularity condition \eqref{eq:regularity} is satisfied with $h=\frac1{16}$, $h_\text{min}=0$, and we obtain a translation invariant defect operator $D(\lambda \psi + \tilde\lambda\bar\psi)$ in the Ising CFT perturbed by $\bm{\sigma}$\cite{Ambrosino:2025myh}.

\subsubsection*{Example: Lee-Yang model}

Consider the perturbation of the Lee-Yang CFT by the bulk field $\bm{\tau}$, i.e.\ $X=\tau=Y$, $\phi_X = |{-}\frac1{5} \rangle = \phi_Y$, and $\omega \neq 0$.
The regularity condition \eqref{eq:regularity} is satisfied with $h=-\frac15$ and $h_\text{min}=-\frac15$.
The only solution to the commutation
with $\eta=0=\tilde\eta$ is the $\mathbf{1}$-defect, so that the $\bm{\tau}$-perturbation preserves no non-trivial topological symmetry. But for the $\tau$-defect we can find a solution with non-zero $\eta,\tilde\eta$, giving a translation invariant defect in the perturbed theory\cite{Runkel:2010ym,Ambrosino:2025myh}.

\medskip

In Ref.\ \citen{Ambrosino:2025myh} we carried out a more expansive search for solutions to the commutation condition in diagonal Virasoro minimal model CFTs. We found generic solutions when the representations $X,Y$ of the perturbing field in Figure~\ref{fig:three-bulk-fields}c are given by the Kac-labels
$(1,2)$, $(1,3)$, or $(1,5)$. These are well-known to be integrable deformations. We also found sporadic additional solutions, such as the $(1,7)$-perturbation of $M(3,10)$ which has recently been investigated in Refs.\,\citen{Katsevich:2024jgq,Ambrosino:2025yug}. However, the additional solutions we found do not satisfy the regularity condition \eqref{eq:regularity} and will require regularisation.

\section{Fusion of translation invariant defects and functional relations}

Fix a 2d\,CFT and a bulk field $\varphi$ determined by $X,Y,\omega$ and $\phi_X$, $\phi_Y$ as in Figure~\ref{fig:three-bulk-fields}c. Write $C(2i\mu\varphi)$ for the 2d\,CFT perturbed by $2i\mu\varphi$ for some $\mu \in \mathbb{C}$,
so that its Hamiltonian is given by \eqref{eq:H0+Hpert}. Suppose that the regularity condition \eqref{eq:regularity} holds.\footnote{This condition is only necessary, we are really assuming the multiple integrals defining the perturbed defects to be finite.}

Solutions $(D,\eta,\tilde\eta)$ to the commutation condition \eqref{eq:commcond} produce perturbative examples of translation invariant defects in $C(2i\mu\varphi)$. Let us denote the category formed by these translation invariant defects by $\mathcal{T}_\text{pert}$, so that \eqref{eq:Ttop-in-T} gets refined to
\begin{equation}\label{eq:Tpert-in-T}
		\mathcal{T}_\text{top} \subset
        \mathcal{T}_\text{pert} \subset \mathcal{T} \ .
\end{equation}
Next we show that $\mathcal{T}_\text{pert}$ is even a monoidal subcategory by describing the tensor product explicitly.

\medskip

Consider two solutions $(D,\eta,\tilde\eta), (E,\zeta,\tilde\zeta)$ of the commutation condition \eqref{eq:commcond}. The corresponding perturbed defects $D(\lambda,\tilde\lambda)$ and $E(\lambda',\tilde\lambda')$ are translation invariant in $C(2i\mu\varphi)$, provided that
$\lambda \tilde\lambda = \mu = \lambda'\tilde\lambda'$. Placing the two defects parallel to each other on a cylinder with distance $r$, by translation invariance any correlator is independent of $r$ (for $r$ small enough, so that no other fields are in the strip between the two defects), and we get a non-singular fusion of the two defects for $r \to 0$.

The fused defect is given by the fused topological defect $D \otimes E$, perturbed by $\psi$ and $\bar\psi$ given by sums of the perturbing fields on $D$ and $E$. For $\psi$ the chiral TFT representation is
\begin{equation}
\includegraphics[valign =c,scale=0.75]{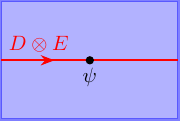} ~~ = ~~ \includegraphics[valign =c,scale=0.75]{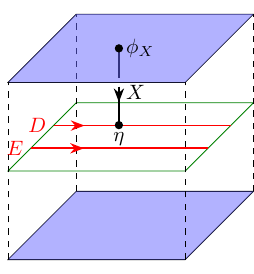} ~ + ~ \includegraphics[valign =c,scale=0.75]{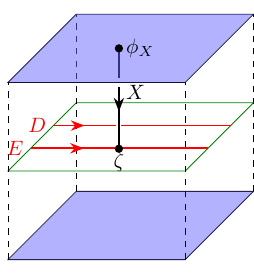}
\end{equation}
and for $\bar\psi$ the picture is analogous.
The fused defect still satisfies the commutation condition. This is easy to check from the explicit expression for the perturbing fields above and the fact that $(D,\eta,\tilde\eta)$ and $(E,\zeta,\tilde\zeta)$ were solutions to start with.
This shows that $\mathcal{T}_\text{pert}$ is closed under fusion and provides us with an explicit description of the tensor product.

\medskip

Next we need the following simple but important sufficient condition for a perturbed defect operator to decompose into a sum of other perturbed defects:
Suppose the underlying topological defect $D$ can be split into a direct sum $D = A \oplus B$ of topological defects in such a way, that the perturbation contains no defect changing field $\Psi_{A \to B}$ from $A$ to $B$ (but possibly one for $B$ to $A$). Then the perturbed defect operator $\widehat D$ decomposes as follows:
\begin{equation}
	\widehat D(\Psi_{A \to A} + \Psi_{B \to B} + \Psi_{B \to A})
    = \widehat A(\Psi_{A \to A}) + \widehat B(\Psi_{B \to B}) \ .
\end{equation}
This holds because the defect $D$ is placed on a cylinder with periodic boundary conditions, and a term in the perturbative expansion with at least one $\Psi_{B \to A}$ is zero, since the counterpart $\Psi_{A \to B}$ is absent. Note that this does not require the defect perturbation to be of the specific form considered in Section~\ref{sec:commcond}, and that that case is recovered for $\Psi = \lambda\psi + \tilde\lambda\bar\psi$.

\medskip

Combining the fusion of defects with the decomposition of defect operators can lead to interesting functional relations which can constrain renormalisation group flows. Below we give examples of such functional relations in the Lee-Yang and Ising CFT in the simplified situation the the bulk CFT remains unperturbed, and that the defects are only perturbed by a holomorphic field, i.e.\ $\omega=\tilde\eta=0$.

\subsubsection*{Example: Lee-Yang model}

Let $D=\tau$ and $X=\tau$, that is, we are considering the $\tau$-defect perturbed by the holomorphic primary of weight $h=-\frac15$. Write $D(\lambda)$ for the corresponding perturbed defect. Upon normalising the defect field appropriately, one finds the following functional relation for the defect operators \cite{Runkel:2007wd,Manolopoulos:2009np}
\begin{equation}
    \widehat D(e^{2\pi i/5} \lambda) \, \widehat D(e^{-2\pi i/5} \lambda)
    \,=\,
    \mathrm{id} + \widehat D(\lambda) \ .
\end{equation}
Using this relation one can constrain the endpoint of the defect flow to be the $\mathbf{1}$-defect
\cite{Manolopoulos:2009np}, in agreement with TCSA and TBA results \cite{Dorey:1997yg}.

\subsubsection*{Example: Ising model}

Consider the case $D=\sigma$ and $X=\epsilon$, that is, the $\sigma$-defect perturbed by the holomorphic primary field of weight $h=\frac12$. If one computes the functional relation according to the formalism presented here, one finds
\begin{equation}
    \widehat D(e^{i \pi/4} \lambda)\,\widehat D(e^{-i \pi/4} \lambda)\,=\,\mathrm{id} + \widehat\epsilon \ .
\end{equation}
However, as already mentioned in Section~\ref{sec:commcond}, this perturbation does not satisfy the regularity condition \eqref{eq:regularity} and needs regularisation. For the corresponding boundary perturbation this has been done exactly in Ref.\ \citen{Chatterjee:1994sv}. For the defect perturbation one finds a correction to the functional relation \cite{Gaiotto:2020fdr},
\begin{equation}
    \widehat D(e^{i \pi/4} \lambda)\,\widehat D(e^{-i \pi/4} \lambda)\,=\,\mathrm{id} + e^{-2 \pi \lambda^2}\widehat\epsilon \ .
\end{equation}
This was used in Ref.\ \citen{Gaiotto:2020fdr} to argue that the $\sigma$-defect flows to the $\mathbf{1}$- or $\epsilon$-defect, depending on the sign of  $\lambda$.

\section{Mathematical description and Yetter-Drinfeld modules}

Here we define a monoidal category $\mathcal{M}_\omega$ which will give a representation theoretic description of $\mathcal{T}_\text{pert}$. We will introduce $\mathcal{M}_\omega$ in a slightly more general setting than the one discussed so far.

Let $\mathcal{C}$ and $\mathcal{D}$ be additive braided monoidal categories, and let $\mathcal{M}$ be an additive monoidal category. In terms of the chiral TFT picture in \eqref{eq:topdef-preserved-chiralTFT}, this amounts to allowing different 3d\,TFTs with different braided categories of line defects above and below the surface defect $\alpha$. The line $X$ in \eqref{eq:topdef-preserved-chiralTFT} would then belong to $\mathcal{C}$ and the line $Y$ to $\mathcal{D}$. The monoidal category $\mathcal{M}$ describes the line defects $D$ localised on the surface defect $\alpha$.

In addition, we fix additive braided monoidal functors $F : \mathcal{C} \to \mathcal{Z}(\mathcal{M})$, $G : \mathcal{D} \to \mathcal{Z}(\mathcal{M})$ where $\mathcal{Z}(\mathcal{M})$ is the Drinfeld centre of $\mathcal{M}$. The images of functors $F$ and $G$ are required to be transparent relative to each other. In more detail, if $\gamma_{FX,-}$ denotes the half-braiding on $F(X)$, and $\delta_{GY,-}$ that on $G(Y)$, we require
\begin{equation} \label{eq:FG-commute}
  \delta_{GY,FX} \circ \gamma_{FX,GY}  = \mathrm{id}_{FX \otimes GY} \ .
\end{equation}
Incidentally, this data defines a 1-morphism in the symmetric monoidal 4-category $\mathrm{BrTens}$ \cite{Brochier:2018kxc}.
In terms of chiral TFT, the functor $F$ encodes how to fuse line defects from $\mathcal{C}$ above the surface defect $\alpha$ into the surface defect, where they turn into the line defect $F(X)$ in $\alpha$, and similarly for $G$, see also Ref.\ \citen{Fuchs:2012dt}.

Denote by $U : \mathcal{Z}(\mathcal{M}) \to \mathcal{M}$ the functor which forgets the half-braiding. As a final piece of data, we fix $X \in \mathcal{C}$, $Y\in\mathcal{D}$ and a morphism
\begin{equation}
	\omega : UF(X) \otimes UG(Y) \to \mathbf{1}
\end{equation}
in $\mathcal{M}$. In the following we will often omit writing the forgetful functor $U$ explicitly.

\begin{definition}
\begin{enumerate}
    \item A triple $(D,\eta,\tilde\eta)$ with
\begin{equation}
	D \in \mathcal{M}
    ~,~~~
	\eta : F(X) \otimes D \to D
    ~,~~~
	\tilde\eta : D \otimes G(Y) \to D
\end{equation}
satisfies the \textit{commutation condition}, if the following identity of morphisms in $\mathcal{M}$ holds (written as string diagrams in $\mathcal{M}$, read from bottom to top):
\begin{equation}\label{eq:commcond-cat}
\includegraphics[scale=0.9, valign = c]{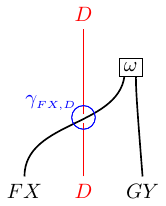} - \includegraphics[scale=0.9, valign = c]{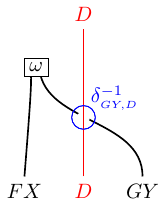}   = \includegraphics[scale=0.9, valign = c]{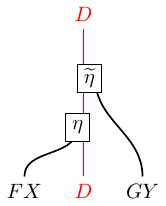}  - \includegraphics[scale=0.9, valign = c]{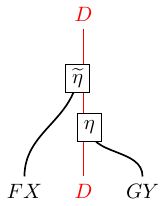}
\end{equation}
\item The category $\mathcal{M}_\omega$ has as objects triples $(D,\eta,\tilde\eta)$ which satisfy the commutation condition. A morphism $f : (D,\eta,\tilde\eta) \to (E,\zeta,\tilde\zeta)$ in $\mathcal{M}_\omega$ is a morphism $f : D \to E$ in $\mathcal{M}$ which satisfies
\begin{equation}
    f \circ \eta = \zeta \circ (\mathrm{id}_{FX} \otimes f)
    \quad \text{and} \quad
    f \circ \tilde\eta = \tilde\zeta \circ (f \otimes \mathrm{id}_{GY}) \ .
\end{equation}
\end{enumerate}
\end{definition}

It is understood that through $\omega$, the category $\mathcal{M}_\omega$ implicitly depends also on the other ingredients $\mathcal{C},\mathcal{D},F,G,X,Y$.
The category $\mathcal{M}_\omega$ is monoidal, with tensor product $\circledast$ defined as
\begin{equation}\label{eq:M-omega-tensor}
     (D,\eta,\tilde\eta) \circledast (E,\zeta,\tilde\zeta)
     \,:=\, \big( D \otimes E , T(\eta,\zeta), \tilde T(\tilde\eta,\tilde\zeta) \big) \ ,
\end{equation}
where $D \otimes E$ is the tensor product in $\mathcal{M}$, and
\begin{equation}
    T(\eta,\zeta) = \hspace{-.2cm}
 \includegraphics[scale=0.9, valign = c]{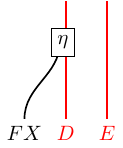} \hspace{-0.05cm} +
 \includegraphics[scale=0.9, valign = c]{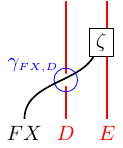}
    ~,~~
    \tilde T(\tilde\eta,\tilde\zeta) =
 \includegraphics[valign = c]{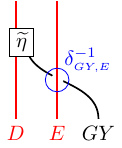} + \includegraphics[valign = c]{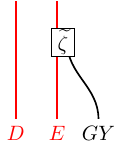}
\end{equation}
One easily checks that the right hand side of \eqref{eq:M-omega-tensor} does again satisfy the commutation condition \eqref{eq:commcond-cat}. The tensor unit is given by $(\mathbf{1},0,0)$.

\medskip

Recall from Section~\ref{sec:commcond} that one can absorb the coupling constant $\mu$ of the bulk perturbation into the topological point junction by replacing $\omega$ by $\mu\omega$. The following proposition states that for non-zero $\mu$, up to equivalence the category of perturbative translation invariant defects does not depend on $\mu$:

\begin{proposition}
Suppose all categories and functors are $\mathbb{C}$-linear. Then for all $\mu \in \mathbb{C}$, $\mu \neq 0$ we have $\mathcal{M}_{\omega} \cong \mathcal{M}_{\mu \omega}$ as additive monoidal categories.
\end{proposition}

\begin{proof}
    It is easy to write explicit functors both ways. Fix $\lambda,\tilde\lambda$ such that $\lambda\tilde\lambda=\mu$. For example, if $(D,\eta,\tilde\eta) \in M_{\omega}$, then $(D,\lambda\eta,\tilde\lambda\tilde\eta) \in M_{\mu\omega}$. It is straightforward to check monoidality of these functors.
\end{proof}

One can check that if $\mathcal{M}$ is pivotal, so is $\mathcal{M}_\omega$, and if $\mathcal{M}$ is abelian, so is $\mathcal{M}_\omega$ \cite{Manolopoulos:2009np}. However, as we will see in the Lee-Yang example below, even if $\mathcal{M}$ is a fusion category, $\mathcal{M}_\omega$ will typically have continuously many simple objects and will no longer be semisimple.

\medskip

Let $\mathcal{M}^\oplus \supset \mathcal{M}$ be the completion of $\mathcal{M}$ with respect to countable direct sums. Then $\mathcal{M}^\oplus$ contains the tensor algebras $T(X) = \mathbf{1} \oplus X \oplus X^{\otimes 2} \oplus X^{\otimes 3} \oplus \dots$ and $T(Y)$, where we abbreviate $X = UF(X)$ and $Y = UG(Y)$. Due to the inherited half-braiding, $T(X)$ and $T(Y)$ are braided Hopf algebras in $\mathcal{M}^\oplus$. The morphism $\omega$ induces a Hopf pairing $\hat\omega : T(X) \otimes T(Y) \to \mathbf{1}$ \cite{Buecher:2012ma}.

A \textit{Yetter-Drinfeld module} in $\mathcal{M}$ is an object $D$ of $\mathcal{M}$ which is both a $T(X)$-left module and a $T(Y)$-right module, but in general not a bimodule. Instead it satisfies \cite{Bespalov1995,Buecher:2012ma}
\begin{equation}
\includegraphics[scale=0.9, valign = c]{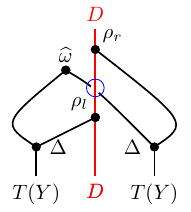}\quad = \quad \includegraphics[scale=0.9, valign = c]{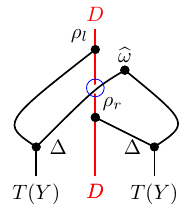}
\end{equation}
Here, $\Delta$ denotes the coproduct of $T(X)$ and $T(Y)$, respectively, and $\rho_\ell$ is the left action of $T(X)$ on $D$, and $\rho_r$ the right action of $T(Y)$ on $D$. The circled crossings are induced by the half braidings of $F(X)$ and $G(Y)$ as in \eqref{eq:commcond-cat}.

Write $\mathrm{YD}_{\mathcal{M},\omega}$ for the category of Yetter-Drinfeld modules. It is always a monoidal category, but it is not necessarily braided. The general formula for the braiding on $\mathrm{YD}_{\mathcal{M},\omega}$ involves the copairing for $\hat\omega$, which does not exist (for $X,Y \neq 0$) as $T(X)$ and $T(Y)$ are not dualisable.

\medskip

In fact, the category $\mathcal{M}_\omega$ is just a different way to talk about Yetter-Drinfeld modules for the Hopf algebras $T(X)$ and $T(Y)$ \cite{Buecher:2012ma}:

\begin{proposition}
Evaluating on $X \subset T(X)$ and $Y \subset T(Y)$ provides a monoidal equivalence $\mathrm{YD}_{\mathcal{M},\omega} \cong \mathcal{M}_\omega$.
\end{proposition}

While $\mathrm{YD}_{\mathcal{M},\omega}$ seems more complicated than $\mathcal{M}_\omega$, this description has its advantages. For example, if one can find Hopf ideals $I \subset T(X)$ and $J \subset T(Y)$, one can single out a monoidal subcategory of $\mathrm{YD}_{\mathcal{M},\omega}$ by restricting to Yetter-Drinfeld modules that descend to the quotients $T(X)/I$ and $T(Y)/J$. These subcategories may have better properties than all of $\mathrm{YD}_{\mathcal{M},\omega}$, such as being almost-everywhere braided. See Ref.\ \citen{Buecher:2012ma} for an example of this.

\subsubsection*{Example: Lee-Yang model}

Let $\mathcal{C}=\mathcal{D}$ be the modular fusion category of representations of the irreducible Virasoro VOA at central charge $c=-22/5$. It has simple objects $\mathbf{1}$, $\tau$ with $\tau \otimes \tau = \mathbf{1} \oplus \tau$. Since there is no non-trivial surface defect in the Lee-Yang model, also $\mathcal{M}=\mathcal{C}$. The functors $F,G$ are given by $F: \mathcal{C} \to \mathcal{Z}(\mathcal{C})$, $F(X) = (X,\gamma_{X,-} = c_{X,-})$, where $c_{-,-}$ denotes the braiding in $\mathcal{C}$, and $G(Y) = (Y,\delta_{Y,-} = c^{-1}_{-,Y})$.
The images of $F$ and $G$ are indeed mutually transparent, as the condition \eqref{eq:FG-commute} now simply reads $c^{-1}_{X,Y} \circ c_{X,Y} = \mathrm{id}_{X \otimes Y}$.

Fix $X=Y=\tau$ and let $\omega : \tau \otimes \tau \to \mathbf{1}$ be any non-zero morphism. One can check that for a suitable choice of $b : \tau \otimes \tau \to \tau$, the triple $(\tau,b,b)$ satisfies the commutation condition \eqref{eq:commcond-cat}, i.e.\ $(\tau,b,b) \in \mathcal{C}_\omega$.
Clearly, for any $\lambda \in \mathbb{C}^\times$, also $(\tau,\lambda b, \lambda^{-1} b) \in \mathcal{C}_\omega$.

Since $\tau$ is simple in $\mathcal{C}$, each $(\tau,\lambda b, \lambda^{-1} b)$ is necessarily simple in $\mathcal{C}_\omega$. And since the only morphisms $\tau \to \tau$ in $\mathcal{C}$ are multiples of the identity, it is easy to check that the $(\tau,\lambda b, \lambda^{-1} b)$ are pairwise non-isomorphic for different values of $\lambda$.

What is more, for generic $\lambda,\rho \in \mathbb{C}^\times$ one finds that the tensor product $(\tau,\lambda b, \lambda^{-1} b) \circledast (\tau,\rho b, \rho^{-1} b)$ is again simple in $\mathcal{C}_\omega$, even though the underlying object $\tau \otimes \tau = \mathbf{1}\oplus\tau$ in $\mathcal{C}$ is not simple. For special ratios $\lambda/\rho$, however, the tensor product is the middle term in a non-split short exact sequence (this follows form a computation analogous to the one in Ref.\ \citen{Manolopoulos:2009np}).

This shows that $\mathcal{C}_\omega$ has an uncountably infinite number of simple objects, even though $\mathcal{C}$ only has two, and that $\mathcal{C}_\omega$ is not semisimple (it has non-split exact sequences), even though $\mathcal{C}$ is.

\section*{Acknowledgments}

IR would like to thank the organisers of the conference
``Non-perturbative Methods in QFT'' at Kyushu University (March 10--14, 2025) for a putting together such an inspiring meeting.
The authors would like to thank Anatoly Konechny for pointing out Refs.\,\citen{Chatterjee:1994sv,Gaiotto:2020fdr}.
The authors thank  the  Deutsche Forschungs\-gemeinschaft (DFG, German Research Foundation) under Germany's Excellence Strategy - EXC 2121 ``Quantum Universe'' - 390833306,  and the Collaborative Research Center - SFB 1624 ``Higher structures, moduli spaces and integrability'' - 506632645, for support. GW was, in addition, supported by the STFC under grant ST/T000759/1.


{\small
\begin{thebibliography}{00}  
\setlength{\itemsep}{0pt}

\bibitem{Ambrosino:2025myh}
F.~Ambrosino, I.~Runkel, G.~M.~T.~Watts,
 {\it J.\ Phys.\ A: Math.\ Theor.\ \textbf{58} (2025) 425401},
\href{https://arxiv.org/abs/2504.05277}{[2504.05277 [hep-th]]}.

\bibitem{Fuchs:2002cm}
J.~Fuchs, I.~Runkel, C.~Schweigert,
{\it Nucl.\ Phys.\ B \textbf{646} (2002) 353--497},
\href{https://arxiv.org/abs/hep-th/0204148}{[hep-th/0204148]}.

\bibitem{Davydov:2011kb}
A.~Davydov, L.~Kong, I.~Runkel,
in {\it Mathematical Foundations of Quantum Field Theory and Perturbative String Theory}, {\it Proc.\ Symp.\ Pure Math.\ \textbf{83} (2011) 71--130},
\href{https://arxiv.org/abs/1107.0495}{[1107.0495 [math.QA]]}.

\bibitem{Chang:2018iay}
C.~M.~Chang, Y.~H.~Lin, S.~H.~Shao, Y.~Wang, X.~Yin,
{\it JHEP \textbf{01} (2019) 026}
\href{https://arxiv.org/abs/1802.04445}{[1802.04445 [hep-th]]}.

\bibitem{Frohlich:2009gb}
J.~Fr\"ohlich, J.~Fuchs, I.~Runkel, C.~Schweigert,
in {\it Proceedings of the XVIth International Congress on Mathematical Physics (2010) 608--613},
\href{https://arxiv.org/abs/0909.5013}{[0909.5013 [math-ph]]}.

\bibitem{Bhardwaj:2017xup}
L.~Bhardwaj, Y.~Tachikawa,
{\it JHEP \textbf{03} (2018) 189},
\href{https://arxiv.org/abs/1704.02330}{[1704.02330 [hep-th]]}.

\bibitem{Gaiotto:2020iye}
D.~Gaiotto, J.~Kulp,
{\it JHEP \textbf{02} (2021) 132},
\href{https://arxiv.org/abs/2008.05960}{[2008.05960 [hep-th]]}.

\bibitem{saclay-proc}
I.~Runkel, in
{\it Proceedings of the 28e rencontre Itzykson}, eds.\ D.~Mazac and S.~Ribault,
\href{https://arxiv.org/abs/2510.09868}{[2510.09868 [hep-th]]}.

\bibitem{Kong:2020cie}
L.~Kong, T.~Lan, X.~G.~Wen, Z.~H.~Zhang, H.~Zheng,
{\it Phys.\ Rev.\ Res.\ \textbf{2} (2020) 043086},
\href{https://arxiv.org/abs/2005.14178}{[2005.14178 [cond-mat.str-el]]}.

\bibitem{Apruzzi:2021nmk}
F.~Apruzzi, F.~Bonetti, I.~Garc{\'\i}a Etxebarria, S.~S.~Hosseini, S.~Sch\"afer-Nameki,
{\it Commun.\ Math.\ Phys.\ \textbf{402} (2023) 895--949},
\href{https://arxiv.org/abs/2112.02092}{[2112.02092 [hep-th]]}.

\bibitem{Freed:2022qnc}
D.~S.~Freed, G.~W.~Moore, C.~Teleman,
{\it Quantum Topol.\ 15 (2024) 779--869},
\href{https://arxiv.org/abs/2209.07471}{[2209.07471 [hep-th]]}.

\bibitem{Kaidi:2022cpf}
J.~Kaidi, K.~Ohmori, Y.~Zheng,
{\it Commun.\ Math.\ Phys.\ \textbf{404} (2023) 1021--1124},
\href{https://arxiv.org/abs/2209.11062}{[2209.11062 [hep-th]]}.

\bibitem{Bhardwaj:2023ayw}
L.~Bhardwaj, S.~Sch\"afer-Nameki,
{\it SciPost Phys.\ 19, 098 (2025)}
\href{https://arxiv.org/abs/2305.17159}{[2305.17159 [hep-th]]}.

\bibitem{Carqueville:2023jhb}
N.~Carqueville, M.~Del Zotto, I.~Runkel,
in {\it Encyclopedia of Mathematical Physics (Second Edition),
Academic Press,
2025, 621--647},
\href{https://arxiv.org/abs/2311.02449}{[2311.02449 [math-ph]]}.

\bibitem{Fredenhagen:2009tn}
S.~Fredenhagen, M.~R.~Gaberdiel, C.~Schmidt-Colinet,
{\it J.\ Phys.\ A \textbf{42} (2009) 495403}
\href{https://arxiv.org/abs/0907.2560}{[0907.2560 [hep-th]]}.

\bibitem{Frohlich:2004ef}
J.~Fr\"ohlich, J.~Fuchs, I.~Runkel, C.~Schweigert,
{\it Phys.\ Rev.\ Lett.\ \textbf{93} (2004) 070601},
\href{https://arxiv.org/abs/cond-mat/0404051}{[cond-mat/0404051]}.

\bibitem{Bazhanov:1994ft}
V.~V.~Bazhanov, S.~L.~Lukyanov, A.~B.~Zamolodchikov,
{\it Commun.\ Math.\ Phys.\ \textbf{177} (1996) 381--398},
\href{https://arxiv.org/abs/hep-th/9412229}{[hep-th/9412229]}.

\bibitem{Konik:1997gx}
R.~Konik, A.~LeClair,
{\it Nucl.\ Phys.\ B \textbf{538} (1999) 587--611},
\href{https://arxiv.org/abs/hep-th/9703085}{[hep-th/9703085]}.

\bibitem{Bachas:2004sy}
C.~Bachas, M.~Gaberdiel,
{\it JHEP \textbf{11} (2004) 065},
\href{https://arxiv.org/abs/hep-th/0411067}{[hep-th/0411067]}.

\bibitem{Runkel:2007wd}
I.~Runkel,
{\it J.\ Phys.\ A \textbf{41} (2008) 105401},
\href{https://arxiv.org/abs/0711.0102}{[0711.0102 [hep-th]]}.

\bibitem{Bachas:2007td}
C.~Bachas, I.~Brunner,
{\it JHEP \textbf{02} (2008) 085},
\href{https://arxiv.org/abs/0712.0076}{[0712.0076 [hep-th]]}.

\bibitem{Bachas:2013ora}
C.~Bachas, I.~Brunner, D.~Roggenkamp,
{\it J.\ Stat.\ Mech.\ \textbf{1308} (2013) P08008},
\href{https://arxiv.org/abs/1303.3616}{[1303.3616 [cond-mat.stat-mech]]}.

\bibitem{Runkel:2010ym}
I.~Runkel,
{\it J.\ Phys.\ A \textbf{43} (2010) 365206},
\href{https://arxiv.org/abs/1004.1909}{[1004.1909 [hep-th]]}.

\bibitem{Buecher:2012ma}
D.~B\"ucher, I.~Runkel,
{\it J.\ Math.\ Phys.\ \textbf{55} (2014) 111705}
\href{https://arxiv.org/abs/1211.4726}{[1211.4726 [math.QA]]}.

\bibitem{Klassen:1990dx}
T.~R.~Klassen, E.~Melzer,
{\it Nucl.\ Phys.\ B \textbf{350} (1991) 635-689}.

\bibitem{Kormos:2009sk}
M.~Kormos, I.~Runkel, G.~M.~T.~Watts,
{\it JHEP \textbf{11} (2009) 057},
\href{https://arxiv.org/abs/0907.1497}{[0907.1497 [hep-th]]}.

\bibitem{Popov:2025cha}
F.~K.~Popov, Y.~Wang,
\href{https://arxiv.org/abs/2504.06203}{[2504.06203 [hep-th]]}.

\bibitem{Katsevich:2024jgq}
A.~Katsevich, I.~R.~Klebanov, Z.~Sun,
{\it JHEP \textbf{03} (2025) 170},
\href{https://arxiv.org/abs/2410.11714}{[2410.11714 [hep-th]]}.

\bibitem{Ambrosino:2025yug}
F.~Ambrosino, S.~Negro,
{\it Phys.\ Rev.\ Lett.\ \textbf{135} (2025) 2}
\href{https://arxiv.org/abs/2501.07511}{[2501.07511 [hep-th]]}.

\bibitem{Manolopoulos:2009np}
D.~Manolopoulos, I.~Runkel,
{\it Commun.\ Math.\ Phys.\ \textbf{295} (2010) 327--362},
\href{https://arxiv.org/abs/0904.1122}{[0904.1122 [hep-th]]}.

\bibitem{Dorey:1997yg}
P.~Dorey, A.~Pocklington, R.~Tateo, G.~Watts,
{\it Nucl.\ Phys.\ B \textbf{525} (1998) 641--663}
\href{https://arxiv.org/abs/hep-th/9712197}{[hep-th/9712197]}.

\bibitem{Chatterjee:1994sv}
R.~Chatterjee,
{\it Mod.\ Phys.\ Lett.\ A \textbf{10} (1995) 973--984}
\href{https://arxiv.org/abs/hep-th/9412169}{[hep-th/9412169]}.

\bibitem{Gaiotto:2020fdr}
D.~Gaiotto, J.~H.~Lee, J.~Wu,
{\it JHEP \textbf{04} (2021) 268},
\href{https://arxiv.org/abs/2003.06694}{[2003.06694 [hep-th]]}.

\bibitem{Brochier:2018kxc}
A.~Brochier, D.~Jordan, N.~Snyder,
{\it Compos.\ Math.\ \textbf{157} (2021) 435--483},
\href{https://arxiv.org/abs/1804.07538}{[1804.07538 [math.QA]]}.

\bibitem{Fuchs:2012dt}
J.~Fuchs, C.~Schweigert, A.~Valentino,
{\it Commun.\ Math.\ Phys.\ \textbf{321} (2013) 543--575},
\href{https://arxiv.org/abs/1203.4568}{[1203.4568 [hep-th]]}.

\bibitem{Bespalov1995}
Y.~N.~Bespalov,
{\it Appl.\ Cat.\ Structures 5 (1997) 155--204},
\href{https://arxiv.org/abs/q-alg/9510013}{[q-alg/9510013]}.

\end{thebibliography}
}
\end{document}